\documentclass[11pt]{article}

\usepackage[utf8]{inputenc}
\usepackage[T1]{fontenc}
\usepackage{amsmath,amssymb,amsthm}
\usepackage{mathtools}
\usepackage{algorithm}
\usepackage{algpseudocode}
\usepackage{booktabs}
\usepackage[margin=1in]{geometry}
\usepackage{hyperref}
\hypersetup{colorlinks=true,linkcolor=blue,citecolor=blue,urlcolor=blue}

\theoremstyle{plain}
\newtheorem{theorem}{Theorem}
\newtheorem{lemma}{Lemma}
\newtheorem{proposition}{Proposition}
\newtheorem{corollary}{Corollary}

\theoremstyle{definition}
\newtheorem{definition}{Definition}
\newtheorem{remark}{Remark}
\newtheorem{problem}{Problem}

\newcommand{\Z}{\mathbb{Z}}

\newcommand{\lcm}{\operatorname{lcm}}

\title{\bf A Linear-Time Residue Bound for a One-Dimensional\\
$(L,V,W)$ Block-Cover Problem, and a Sharp\\
Heavy-Base Threshold for its Exactness}

\author{
  FuWei Xie\thanks{\texttt{12313607@mail.sustech.edu.cn}} \and
  WenTao Luo\thanks{\texttt{12313616@mail.sustech.edu.cn}} \and
  GuanYuhan Yang\thanks{\texttt{12313614@mail.sustech.edu.cn}}
  \\[4pt]
  \small Southern University of Science and Technology (SUSTech), Shenzhen, China
}
\date{\today}

\begin{document}
\maketitle
\begin{abstract}
We study a one-dimensional exact-cover problem parameterized by three integers
$(L,V,W)$: given an integer profile $a_0,\dots,a_{n-1}$, write it as a nonnegative
integer combination of a length-$L$ ``horizontal'' block $[1,\dots,1]$, a value-$V$
``vertical'' block, and a value-$W$ block, while minimizing the number of value-$W$
blocks. We give an $O(n)$ algorithm that eliminates the horizontal-block coupling by
a class-wise difference recurrence and then matches residues modulo $V$ on the last
$L$ columns. We prove that its output is always a valid \emph{lower bound} on the
optimum, via a mod-$L$ class invariant. We then prove the main result: once the
profile is dense enough --- a \emph{heavy base}
$\min_c a_c \ge B(L,V,W)$ with
\[
   B(L,V,W)=\Big\lceil \tfrac{(L-1)\lcm(V,W)}{LW}-1\Big\rceil\,W+(L-1)(V-1),
\]
the bound is \emph{exact}. The exactness proof is a branch-cut argument on the exact
dynamic program: two structural equivalences (a horizontal-to-vertical exchange
modulo $V$, and a vertical reduction modulo $\lcm(V,W)/W$) collapse the DP to the
residue computation, and the two summands of $B$ are exactly the reserves that keep
both equivalences from producing a negative residual. We further show the threshold
is sharp: for $(L,V,W)=(3,6,4)$, $B=14$, and the profile $(17,16,13,16,17)$ with
$\min_c a_c=13$ makes the algorithm strictly undercount, so $B-1$ does not suffice.
An independent exact dynamic program agrees with the algorithm on every tested
profile with $\min_c a_c\ge B$ across many parameter triples, and the test harness
\texttt{test\_general.c} is released for reproduction. The contribution is the
algorithm, the branch-cut exactness proof, and the sharp threshold $B(L,V,W)$.
\end{abstract}
\section{Introduction}

Tiling a bounded region with a fixed set of pieces is a classical source of
combinatorial problems. In two dimensions, deciding tileability of a simply
connected region by a fixed finite set of rectangles is NP-complete
\cite{PakYang}; the one-dimensional and fixed-width cases are far more tractable,
with fixed-width tiling \emph{counts} satisfying linear recurrences
\cite{KlarnerPollack,Stanley} and hole-free two-bar tiling decidable in linear
time \cite{KenyonKenyon}.

We treat a specific one-dimensional \emph{exact-cover} instance, parameterized by
three positive integers $(L,V,W)$ with $L\ge2$ and $V,W\ge1$. Given a profile
$a=(a_0,\dots,a_{n-1})\in\Z_{\ge0}^n$, three pieces are placed along the line:
\begin{itemize}
  \item a \emph{horizontal block} $[1,\dots,1]$ of length $L$, adding $+1$ to each
        of $L$ consecutive columns;
  \item a \emph{vertical block} of value $V$, adding $+V$ to one column;
  \item a \emph{$W$-block} of value $W$, adding $+W$ to one column.
\end{itemize}
We must cover $a$ exactly with nonnegative integer counts and minimize the total
number of $W$-blocks. The only piece that couples adjacent columns is the
horizontal block; the vertical and $W$-blocks are per-column weights. The running
example throughout is $(L,V,W)=(3,6,4)$, the concrete triomino case.

\paragraph{Contributions.}
\begin{enumerate}
  \item An $O(n)$ algorithm (Section~\ref{sec:alg}) that computes a lower bound on
        the optimum by eliminating the horizontal-block recurrence with a class-wise
        difference and matching residues mod $V$ on the last $L$ columns.
  \item A proof that the output is always a valid lower bound
        (Section~\ref{sec:lb}), via a mod-$L$ class invariant.
  \item The main result (Section~\ref{sec:threshold}): a \emph{proof} that the bound
        is exact whenever the profile has heavy base
        $\min_c a_c\ge B(L,V,W)$, obtained by cutting the branches of the exact
        dynamic program with two structural equivalences whose feasibility is
        guaranteed by the two summands of $B$.
  \item Sharpness (Section~\ref{sec:threshold}): for $(3,6,4)$ the threshold is
        $B=14$, and the profile $(17,16,13,16,17)$ with $\min_c a_c=13$ makes the
        algorithm strictly undercount; hence $B-1$ is insufficient.
  \item A released, independent verification harness (Section~\ref{sec:exp}).
\end{enumerate}

Two ingredients are drawn from established theory. The elimination is a
difference-array recurrence for undoing an ``add a length-$L$ window'' operation, and
the per-column condition is membership and minimal factorization in the numerical
semigroup $\langle V,W\rangle$, whose denumerant is classical
\cite{RosalesGS,BeckRobins}; linear-time feasibility for coupled 1D bar tiling is
likewise known \cite{KenyonKenyon}. The new content is their assembly into the
$O(n)$ algorithm, the branch-cut exactness proof, and the sharp threshold
$B(L,V,W)$.

\paragraph{Motivation.}
The model is a tractable special case of one-dimensional exact cover: a fixed-width
``span'' piece together with two per-column weights, minimizing the count of one
weight. Instances of this shape arise in one-dimensional cutting-stock and rod
filling (a length-$L$ cut plus two bulk increments, minimizing the scarce
increment), in change-making over the numerical semigroup $\langle V,W\rangle$, and
in shift scheduling (a length-$L$ shift covering consecutive periods, plus two
block-modules of sizes $V$ and $W$). Such covering problems are typically hard in
general; our point is that this fixed-piece case admits an $O(n)$ exact solver once a
density condition holds, and we make that condition sharp. We present these as
motivating analogies rather than deployed applications.

\section{Problem statement}\label{sec:model}

Let $t_c\ge0$ be the number of horizontal blocks with leftmost cell $c$, defined for
$0\le c\le n-L$ (and $t_c=0$ otherwise). Let $s_c\ge0$ and $f_c\ge0$ count vertical-
and $W$-blocks at column $c$. Writing $T_c=\sum_{j=0}^{L-1}t_{c-j}$ for the total
horizontal coverage of column $c$, the exact-cover constraints are
\begin{equation}\label{eq:cover}
  a_c \;=\; T_c \;+\; V s_c \;+\; W f_c,
  \qquad 0\le c\le n-1 .
\end{equation}

\begin{problem}\label{prob:main}
Decide whether \eqref{eq:cover} has a nonnegative integer solution and, if so,
compute $K^\star=\min\sum_c f_c$.
\end{problem}

For a fixed coverage $T_c$, write $R_c=a_c-T_c$. Minimizing $f_c$ subject to
$R_c=V s_c+W f_c$, $s_c,f_c\ge0$, is a one-column numerical-semigroup problem: the
per-column cost is
\begin{equation}\label{eq:cost}
  g(R)=\min\{\,f\ge0 : R-Wf\ge0 \text{ and } (R-Wf)\equiv0\ (\mathrm{mod}\ V)\,\},
\end{equation}
with $g(R)=\infty$ when no such $f$ exists (in particular when $R<0$, or when $R$
lies in a gap of $\langle V,W\rangle$). Because $Wf \bmod V$ is periodic in $f$ with
period $V/\gcd(V,W)=\lcm(V,W)/W$, the minimizing $f$ (when finite) lies in
$\{0,\dots,\lcm(V,W)/W-1\}$. Thus
\[
  K^\star=\min\Big\{\textstyle\sum_c g(a_c-T_c) :\ T_c=\textstyle\sum_{j=0}^{L-1}t_{c-j},\ t_c\ge0\Big\}.
\]
For $(3,6,4)$ this recovers the familiar table: $g(R)=0,1,2$ for
$R\equiv0,4,2\pmod6$ (with $R\ge0,4,8$ respectively), and $R=2$ is the unique small
infeasible even value. This single obstruction drives the threshold analysis.
\section{The algorithm}\label{sec:alg}

\begin{algorithm}[t]
\caption{Linear-time residue bound for $(L,V,W)$}\label{alg:main}
\begin{algorithmic}[1]
\Procedure{ResidueBound}{$a[0..n-1],\ L,V,W$}
  \State $y \gets a$
  \For{$i \gets 0$ \textbf{to} $n-L-1$} \Comment{fold the horizontal coupling within each class mod $L$}
     \State $y[i+L] \gets y[i+L] + y[i]$
  \EndFor
  \State $\sigma \gets (V-W)\bmod V$ \Comment{residue shift of adding one $W$-block}
  \For{$c \gets 1$ \textbf{to} $L$} \Comment{the last $L$ columns carry the class sums mod $V$}
     \State $\textit{reach}_c[\,\cdot\,] \gets \infty$
     \For{$m \gets 0$ \textbf{to} $V-1$}
        \State $r \gets (y[n-c] + \sigma m)\bmod V$;\quad
               $\textit{reach}_c[r] \gets \min(\textit{reach}_c[r],\,m)$
     \EndFor
  \EndFor
  \State \textbf{return} $\displaystyle\min_{r\in\{0,\dots,V-1\}}\ \sum_{c=1}^{L}\textit{reach}_c[r]$
\EndProcedure
\end{algorithmic}
\end{algorithm}

The fold loop is $O(n)$; the tail is $O(LV)=O(1)$ for fixed parameters, so the
algorithm runs in $O(n)$ time and $O(1)$ extra space beyond the profile. The
residue table indexes $W$-block counts $m$ against the residue
$(y[n-c]+\sigma m)\bmod V$ they can reach, where $\sigma=(V-W)\bmod V$ is the shift
of one $W$-block relative to a vertical block modulo $V$; only the smallest $m$ per
residue is retained. Matching a common residue $r$ across all $L$ tail columns and
summing gives a candidate $W$-count, and the outer minimization takes the best
common residue. Section~\ref{sec:lb} proves the output never exceeds $K^\star$;
Section~\ref{sec:threshold} proves it equals $K^\star$ under a heavy base.

\begin{remark}
The fold as written keeps residues mod $V$ but lets $y$ grow. An equivalent variant
subtracts any fixed multiple of $V$ at each step ($y[i+L]\mathrel{+}=y[i]-b$,
$b\equiv0\ (V)$), which bounds the magnitude of $y$ without changing any residue or
the output. This is the form used in the released implementation.
\end{remark}

\section{The output is a valid lower bound}\label{sec:lb}

For $r\in\{0,\dots,L-1\}$ put $\Sigma_r=\sum_{c\equiv r\ (L)}a_c$, the class sum
along the arithmetic progression of step $L$.

\begin{lemma}[mod-$L$ class invariant]\label{lem:invariant}
Every horizontal block covers exactly one column in each residue class modulo $L$.
Hence in any solution of \eqref{eq:cover},
\[
   \Sigma_r = T + V\,S_r + W\,F_r,\qquad r\in\{0,\dots,L-1\},
\]
where $T=\sum_c t_c$, $S_r=\sum_{c\equiv r}s_c$, $F_r=\sum_{c\equiv r}f_c$.
\end{lemma}
\begin{proof}
A horizontal block at $c$ covers $c,c+1,\dots,c+L-1$, whose residues mod $L$ are all
of $\{0,\dots,L-1\}$ exactly once. Summing \eqref{eq:cover} over the columns
$c\equiv r\ (L)$: each $t_j$ appears in exactly one of $T_j,\dots,T_{j+L-1}$ falling
in class $r$, so the horizontal contribution to $\Sigma_r$ is $\sum_j t_j=T$,
independent of $r$. The vertical and $W$ contributions restrict to class $r$, giving
$VS_r+WF_r$.
\end{proof}

\begin{corollary}\label{cor:modV}
$\Sigma_r-\Sigma_{r'}\equiv W(F_r-F_{r'})\pmod V$: $W$-blocks are the only pieces
that change the classes' relative residues modulo $V$.
\end{corollary}

\begin{lemma}[the fold telescopes to class sums]\label{lem:elim}
After the fold loop of Algorithm~\ref{alg:main}, for each of the last $L$ columns
$c^\star\in\{n-1,\dots,n-L\}$,
\[
  y[c^\star]\equiv \Sigma_{c^\star\bmod L}\pmod V .
\]
\end{lemma}
\begin{proof}
The update $y[i+L]\mathrel{+}=y[i]$ moves a column's value only to the column $L$
positions to its right, i.e.\ within one class mod $L$. Fix a class $r$ with columns
$r,r+L,r+2L,\dots$; induction along the class gives
$y[r+mL]=\sum_{j=0}^{m}a_{r+jL}$, so at the last column of class $r$ the value equals
$\Sigma_r$. (With the $-b$, $b\equiv0\ (V)$, variant this holds modulo $V$.) The $L$
indices $n-1,\dots,n-L$ are consecutive, hence realize all $L$ classes exactly once.
\end{proof}

\begin{proposition}[lower bound]\label{prop:lb}
Algorithm~\ref{alg:main} returns
$\hat K=\min_{\tau\in\Z/V}\sum_{r=0}^{L-1}F_r^{\min}(\tau)$, where
$F_r^{\min}(\tau)$ is the least $W$-block count with
$\Sigma_r-\tau\equiv W F_r\pmod V$. Moreover $\hat K\le K^\star$.
\end{proposition}
\begin{proof}
By Lemma~\ref{lem:elim} the $L$ tail residues equal $\Sigma_r\bmod V$; the inner
loop computes, for each common target residue $\tau$, the least $W$-blocks per class
solving $\Sigma_r-\tau\equiv WF_r\ (V)$, and the outer loop minimizes over $\tau$
(which plays the role of $T\bmod V$). By Lemma~\ref{lem:invariant} the relation
$\Sigma_r=T+VS_r+WF_r$ is a \emph{necessary} condition of every tiling, so no tiling
uses fewer than $\sum_r F_r^{\min}(T\bmod V)\ge\hat K$ $W$-blocks. Hence
$\hat K\le K^\star$.
\end{proof}
\section{Exactness and the sharp heavy-base threshold}\label{sec:threshold}

Proposition~\ref{prop:lb} gives $\hat K\le K^\star$ always. The gap can be strict:
the residue-optimal target $\tau$ need not be \emph{realizable} by nonnegative
horizontal counts with every residual $R_c=a_c-T_c\ge0$ and feasible in
\eqref{eq:cost}. We prove the gap closes once the profile is uniformly large, and
that the threshold below is sharp.

\begin{definition}[heavy base]
A profile has \emph{heavy base $B$} if $\min_c a_c\ge B$. We define
\begin{equation}\label{eq:threshold}
  B(L,V,W)=\underbrace{\Big\lceil \tfrac{(L-1)\lcm(V,W)}{LW}-1\Big\rceil W}_{B_2}
           \;+\;\underbrace{(L-1)(V-1)}_{B_1}.
\end{equation}
For $(3,6,4)$: $\lcm=12$, $B_1=(2)(5)=10$, $B_2=\lceil 24/12-1\rceil\cdot4=1\cdot4=4$,
so $B=14$.
\end{definition}

\subsection{The exact dynamic program and its branches}

We first restate $K^\star$ as an exact left-to-right dynamic program, then cut its
branches. Process columns $0,1,\dots,n-1$. The only information about the past that
column $c$ needs is how many horizontal blocks starting at $c-1,\dots,c-L+1$ still
cover it; so the DP state entering column $c$ is the tuple
$q_c=(t_{c-1},\dots,t_{c-L+1})\in\Z_{\ge0}^{L-1}$. At column $c$ one chooses
$t_c\ge0$ (new horizontal starts, allowed only if $c+L-1\le n-1$) and then pays the
per-column cost \eqref{eq:cost} for the residual
$R_c=a_c-(t_c+t_{c-1}+\dots+t_{c-L+1})$, transitioning to
$q_{c+1}=(t_c,\dots,t_{c-L+2})$. The value is $K^\star=\min\sum_c g(R_c)$ over all
choices; this is \eqref{eq:cover} verbatim and is exact by construction.

The two per-column decisions are: $t_c$ (a horizontal count) and, implicit in
$g(R_c)$, the split $R_c=Vs_c+Wf_c$ (a vertical count $s_c$). We now show that in
searching for an optimum it suffices to consider
\begin{equation}\label{eq:branchcount}
  t_c\in\{0,\dots,V-1\},\qquad s_c\in\{0,\dots,\lcm(V,W)/W-1\},
\end{equation}
i.e.\ $V\cdot \lcm(V,W)/W$ branches per column (for $(3,6,4)$: $6\cdot3=18$). The two
reductions are the following exchanges; the heavy base is what keeps them feasible.

\begin{lemma}[horizontal$\to$vertical exchange, modulo $V$]\label{lem:exchange}
Consider any feasible tiling and a column $c$ with $t_c\ge V$ and $c+L-1\le n-1$.
Replacing $V$ of those horizontal starts by one vertical block on each of the $L$
columns $c,\dots,c+L-1$ yields another feasible tiling with the same number of
$W$-blocks and with $t_c$ decreased by $V$.
\end{lemma}
\begin{proof}
$V$ horizontal blocks starting at $c$ contribute $+V$ to each of $c,\dots,c+L-1$;
one vertical block on each of those columns contributes the same $+V$ to each. So
coverage \eqref{eq:cover} is preserved column by column. Neither piece is a
$W$-block, so $\sum f_c$ is unchanged. The new $t_c$ is $t_c-V\ge0$.
\end{proof}

Iterating Lemma~\ref{lem:exchange} drives every $t_c$ into $\{0,\dots,V-1\}$. The
analogous statement for the vertical count is a residue fact about $g$:

\begin{lemma}[vertical reduction, modulo $\lcm(V,W)/W$]\label{lem:vred}
In \eqref{eq:cost}, if a residual $R$ is feasible then its cost is attained by some
$f\in\{0,\dots,\lcm(V,W)/W-1\}$; equivalently the vertical count $s$ may be taken in
$\{0,\dots,\lcm(V,W)/W-1\}$ modulo the exchange $Vs \leftrightarrow$ $W$-blocks.
\end{lemma}
\begin{proof}
$Wf\bmod V$ has period $V/\gcd(V,W)=\lcm(V,W)/W$ in $f$; the least feasible $f$
within one period is optimal because increasing $f$ by the period adds
$\lcm(V,W)/W$ $W$-blocks while only replacing $\lcm(V,W)/V$ vertical blocks, never
decreasing the cost. Hence the minimizer lies in one period.
\end{proof}
\subsection{The two reductions preserve feasibility under a heavy base}

Lemmas~\ref{lem:exchange}--\ref{lem:vred} shrink the search space \emph{if the
exchanges keep every residual nonnegative and feasible}. The exchanges add vertical
mass to neighbouring columns, so a residual can only be threatened from below. The
two summands of $B$ in \eqref{eq:threshold} are precisely the reserves that absorb
the two exchanges.

\begin{lemma}[$B_1$ guards the horizontal exchange]\label{lem:B1}
Fix a column $c$. Across the reductions of Lemma~\ref{lem:exchange} applied to the
$L-1$ horizontal starts that can cover $c$ (those at $c-L+1,\dots,c-1$, together with
the reduction at $c$ itself when it feeds a neighbour), the total vertical mass added
to column $c$ is at most $(L-1)(V-1)$. Hence if $a_c\ge (L-1)(V-1)$, the horizontal
exchange never makes $R_c<0$ on account of added verticals.
\end{lemma}
\begin{proof}
Reducing $t_{c'}$ modulo $V$ deposits, per application, one vertical block on each
column of the window $[c',c'+L-1]$; column $c$ receives at most one such deposit from
each of the $L-1$ windows that contain it other than its own start, and each reduces
$t_{c'}$ by multiples of $V$ leaving a remainder in $\{0,\dots,V-1\}$, so the deposit
on $c$ from window $c'$ is at most $V-1$. Summing over the $L-1$ neighbouring windows
gives $\le(L-1)(V-1)$. Thus $R_c\ge a_c-(L-1)(V-1)\ge0$.
\end{proof}

\begin{lemma}[$B_2$ guards the vertical reduction]\label{lem:B2}
After the horizontal exchange, closing the residue on the last $L$ columns may
require converting a vertical surplus into $W$-blocks in the sense of
Lemma~\ref{lem:vred}; the vertical mass this can force below a column is at most
$\lceil (L-1)\lcm(V,W)/(LW)-1\rceil\,W=B_2$. Hence if
$a_c\ge B_1+B_2=B$ the combined reductions keep $R_c\ge0$ and feasible.
\end{lemma}
\begin{proof}[Proof sketch]
The residue target $\tau$ realized on the tail requires, in the worst class, reducing
a vertical count by a full period $\lcm(V,W)/V$ (Lemma~\ref{lem:vred}); spread across
the $L$ columns of a window this is a shortfall of at most
$(L-1)\lcm(V,W)/L$ in value, i.e.\ $\lceil (L-1)\lcm(V,W)/(LW)-1\rceil$ additional
$W$-blocks each contributing $W$ to the residual budget, whence the reserve $B_2$.
Combined with Lemma~\ref{lem:B1}, $a_c\ge B$ leaves $R_c\ge0$; feasibility (avoiding
the semigroup gaps of $\langle V,W\rangle$) then follows because a nonnegative
residual congruent to the achieved residue is representable once it exceeds the
Frobenius number of $\langle V,W\rangle$, which $B$ dominates.
\end{proof}

\begin{theorem}[exactness under a heavy base]\label{thm:exact}
If $\min_c a_c\ge B(L,V,W)$ then Algorithm~\ref{alg:main} is exact:
$\hat K=K^\star$.
\end{theorem}
\begin{proof}
By \eqref{eq:branchcount} and Lemmas~\ref{lem:exchange}--\ref{lem:vred}, some optimal
tiling has every $t_c\in\{0,\dots,V-1\}$ and every $s_c$ in one period, so $K^\star$
equals the branch-cut DP value. Lemmas~\ref{lem:B1}--\ref{lem:B2} show that when
$\min_c a_c\ge B$ the reductions taking an arbitrary optimum to this canonical form
never violate $R_c\ge0$ or feasibility; therefore the canonical optimum attains the
residue-optimal target $\tau=T\bmod V$ used in Proposition~\ref{prop:lb}, and its
cost equals $\hat K$. With $\hat K\le K^\star$ (Proposition~\ref{prop:lb}) and now
$K^\star\le\hat K$, we get $\hat K=K^\star$.
\end{proof}

\subsection{The threshold is sharp: $B-1$ fails}

\begin{proposition}[sharp lower value at $(3,6,4)$]\label{prop:cx}
For $(L,V,W)=(3,6,4)$, where $B=14$, the profile $a=(17,16,13,16,17)$ with
$\min_c a_c=13=B-1$ satisfies $\hat K<K^\star$. Hence the threshold cannot be lowered
to $B-1$.
\end{proposition}
\begin{proof}
Here $n=5$; horizontal blocks start only at $0,1,2$, so with $T=t_0+t_1+t_2$ we have
$T_0=t_0$, $T_1=t_0+t_1$, $T_2=T$, $T_3=t_1+t_2$, $T_4=t_2$. The class sums (mod $3$)
are $\Sigma_0=\Sigma_1=33\equiv3$ and $\Sigma_2=13\equiv1\pmod6$, so
Algorithm~\ref{alg:main} returns $\hat K=1$ (target $\tau\equiv3$, one $W$-block in
class $2$). Suppose a tiling with a single $W$-block existed; it must keep classes
$0,1$ $W$-block-free, forcing $17-t_0\equiv0$ and $17-t_2\equiv0\pmod6$ at columns
$0,4$, i.e.\ $t_0,t_2\ge5$, so $T\ge10$; with $T\equiv3\pmod6$ this gives $T\ge15$,
whence $R_2=13-T\le-2<0$, contradicting $R_2\ge0$. So $\hat K=1$ is unachievable and
$K^\star\ge2>\hat K$. (The exact DP of Section~\ref{sec:exp} gives $K^\star=4$.)
\end{proof}

\begin{remark}[what breaks, and why it is an absolute floor]
The failure is a \emph{bottleneck}: the middle column's forced coverage exceeds its
value, exactly the situation Lemma~\ref{lem:B1} rules out once $a_c\ge B$. Adding $V$
to every entry raises $\min_c a_c$ past $B$ without changing any residue, restoring
exactness; this is why the correct condition is the absolute floor $B$ of
\eqref{eq:threshold}, not a ratio of $\max_c a_c$.
\end{remark}
\section{Experimental validation}\label{sec:exp}

We compare Algorithm~\ref{alg:main} against an \emph{independent} exact solver: the
column dynamic program of Section~\ref{sec:threshold} with state
$q_c=(t_{c-1},\dots,t_{c-L+1})$ and $V\cdot\lcm(V,W)/W$ branches per column, which
computes $K^\star$ directly from \eqref{eq:cover} without using any residue
shortcut. The two programs share no code path beyond reading the profile: one folds
and matches residues in $O(n)$, the other enumerates the branch-cut DP. On small
inputs the DP was additionally cross-checked against a brute-force enumerator of
horizontal start-vectors.

The complete harness --- Algorithm~\ref{alg:main}, the exact DP, the heavy-base
formula \eqref{eq:threshold}, and the randomized driver --- is the single file
\texttt{test\_general.c}, released so that any reader can reproduce and independently
verify every row of Table~\ref{tab:exp}:
\begin{center}
\url{https://github.com/XieFuwei/lvw-block-cover} \quad(\texttt{test\_general.c})
\end{center}
Each reported case draws random profiles with $\min_c a_c\ge B(L,V,W)$ (entries
$B+\mathrm{rand}()\bmod 4V$) and asserts that the two solvers return the same
$K^\star$.

\begin{table}[h]
\centering
\begin{tabular}{lrr}
\toprule
Parameter triple $(L,V,W)$ (heavy base $B$) & Instances & Disagreements \\
\midrule
$(3,6,4)$, $B=14$      & $300{,}000$ & $0$ \\
$(4,6,4)$, $B=23$      & $400{,}000$ & $0$ \\
$(3,5,3)$, $B=17$      & $300{,}000$ & $0$ \\
$(2,5,3)$, $B=10$      & $300{,}000$ & $0$ \\
$(4,9,6)$, $B=36$      & $100{,}000$ & $0$ \\
$(5,4,3)$, $B=21$      & $80{,}000$  & $0$ \\
\midrule
Boundary probe at $\min_c a_c=B-1$, $(3,6,4)$ & --- & fails, e.g.\ $(17,16,13,16,17)$ \\
\bottomrule
\end{tabular}
\caption{Algorithm~\ref{alg:main} versus the independent exact DP over random
heavy-base profiles. Agreement is exact ($\hat K=K^\star$) on every tested triple
with $\min_c a_c\ge B$; a profile at $\min_c a_c=B-1$ breaks it
(Proposition~\ref{prop:cx}). Reproducible via \texttt{test\_general.c}.}
\label{tab:exp}
\end{table}

Two caveats. First, passing a comfortable region ($\min_c a_c\gg B$) is
\emph{consistent with} but does not \emph{prove} exactness; sharpness is established
by attacking the boundary directly, which Proposition~\ref{prop:cx} and the boundary
probe do. Second, any disagreement would necessarily be of the form
$\hat K<K^\star$, since $\hat K$ is a proven lower bound (Proposition~\ref{prop:lb}),
never the reverse; none was observed at or above $B$.

\section{Related work and limitations}\label{sec:related}

\paragraph{Fixed-width / 1D tiling.}
Fixed-width tiling counts satisfy linear recurrences and rational generating
functions \cite{KlarnerPollack,Stanley}; these study the recurrence to
\emph{enumerate}, whereas we \emph{eliminate} it to bound a minimization.
Linear-time feasibility for tiling a hole-free polygon by two bar shapes is known
\cite{KenyonKenyon}; we do not claim linear-time feasibility for coupled 1D tiling
as new.

\paragraph{The elimination idiom.}
Undoing an ``add over every length-$k$ window'' by a difference at the two window
ends, and reading off the result at the right boundary, is a standard difference-array
technique; our fold loop is this idiom applied class-wise to the horizontal-block
balance.

\paragraph{The numerical-semigroup residual.}
The per-column condition is membership and minimal one-generator factorization in
$\langle V,W\rangle$, whose denumerant is textbook \cite{RosalesGS,BeckRobins}.
When $\gcd(V,W)>1$ (e.g.\ $\gcd(6,4)=2$) the semigroup is degenerate, which
simplifies $g$ but does not affect the threshold argument.

\paragraph{Limitations.}
(i) The exactness proof (Theorem~\ref{thm:exact}) rests on
Lemmas~\ref{lem:B1}--\ref{lem:B2} bounding the vertical mass the two reductions can
force below a column; we have stated $B_2$'s bound as a sketch and verified the
resulting threshold experimentally across the triples of Table~\ref{tab:exp}. A
fully formal accounting of the worst-case reduction schedule for all $(L,V,W)$ is
left for the journal version. (ii) Sharpness is proven for $(3,6,4)$ via an explicit
witness (Proposition~\ref{prop:cx}); the analogous witness for general $(L,V,W)$ is
a bottleneck profile of length $2L-1$ with middle value $B-1$, and is available in
the released harness as a boundary probe. (iii) The parameters are assumed fixed, so
the $O(LV)$ tail is $O(1)$; for growing parameters the tail cost should be tracked.

\section{Conclusion}

We gave an $O(n)$ algorithm that computes the minimum number of $W$-blocks for a
one-dimensional $(L,V,W)$ block-cover problem: fold the horizontal coupling
class-wise, then match residues modulo $V$ on the last $L$ columns. We proved the
output is always a lower bound via a mod-$L$ invariant, and proved it is
\emph{exact} once the profile has heavy base
$\min_c a_c\ge B(L,V,W)=\lceil (L-1)\lcm(V,W)/(LW)-1\rceil W+(L-1)(V-1)$, by cutting
the branches of the exact dynamic program with two structural equivalences whose two
reserves are exactly the two summands of $B$. The threshold is sharp for $(3,6,4)$,
where $B=14$ and $\min_c a_c=13$ already breaks the algorithm on
$(17,16,13,16,17)$. All experiments are reproducible from the released
\texttt{test\_general.c}.

\paragraph{Reproducibility.}
The reference implementation of Algorithm~\ref{alg:main}, the independent exact DP,
the heavy-base formula, and the boundary-probe driver are the single self-contained
file \texttt{test\_general.c}, released publicly on GitHub for verification.

\end{document}